\documentclass[preprint,12pt]{elsarticle}

\usepackage{amssymb,amsmath}

\newtheorem{theorem}{Theorem}

\newtheorem{lemma}[theorem]{Lemma}
\newtheorem{corollary}[theorem]{Corollary}
\newdefinition{remark}{Remark}
\newproof{proof}{Proof}

\journal{arXiv}

\begin{document}

\begin{frontmatter}



\title{Partial classification of cuspidal simple modules for Virasoro-like algebra\tnoteref{titlenote}}
\tnotetext[titlenote]{Supported by the National Natural Science Foundation of China (No.11171294).}


\author[Jiang]{Jian-Jian Jiang\corref{correspondingauthor}}
\cortext[correspondingauthor]{Corresponding author.}
\address[Jiang]{Chern Institute of Mathematics, Nankai University, Tianjin 300071, P. R. China}
\ead{j.j.jiang@foxmail.com}

\author[Lin]{Wei-Qiang Lin}
\address[Lin]{Department of Mathematics and Statistics, Minnan Normal University, Zhangzhou 363000, Fujian, P. R. China}

\begin{abstract}
Let $\mathfrak P$ be the Lie algebra of Hamiltonian vector fields on the torus, which is also known as the Virasoro-like algebra, a special kind of the so-called Block type Lie algebra. And let $\mathfrak A$ be the Laurent polynomial algebra in two variables. In this paper, by following S.E. Rao's strategy of ``backward induction'', we prove that any quasi-finite simple $(\mathfrak A,\mathfrak P)$-module has to come from Larsson-Shen's construction.
\end{abstract}


\end{frontmatter}



\section{Introduction}
\label{introduction}

The group of diffeomorphisms of a manifold plays a significant role in both mathematics and physics. However, it usually happens that such group is quite large and complicated. So one is forced to study its infinitesimal counterpart, namely, the Lie algebra of vector fields on the same manifold. Among others, we are especially interested in the $d$-dimensional torus, due to the fact that they are relatively simple and appear in a lot of fields. It is worth pointing out that the $d=1$ case is associated with the famous Virasoro algebra, which is investigated by both mathematicians and physicists in great detail. The Lie algebras of vector fields on higher-dimensional tori have also attracted many researchers' attention.

Apart from the full algebra of vector fields, one may also consider its notable subalgebra consisting of Hamiltonian vector fields, provided that the dimension $d$ is even. This subalgebra comes from the natural Poisson structure on every even-dimensional torus, which is a fundamental object in theoretical mechanics and symplectic geometry. When $d=2$, the Lie algebra of Hamiltonian vector fields is known to be the Virasoro-like algebra (was firstly introduced in \cite{kirkman}) in the literature on Lie theory. Numerous papers are dedicated to the study of this algebra.

The representation theory of these algebras usually sheds light on the exploration of their structure and geometric properties. A thorough understanding of the representation theory of Virasoro algebra has been achieved. For Lie algebras of vector fields on higher-dimensional tori, some remarkable progresses have been made in the last few years, especially the complete classification of quasi-finite simple modules by Y. Billig and V. Futorny \cite{billig}. As for Virasoro-like algebra, although some interesting modules have been found, and all highest weight type modules have been classified, we still do not know whether there are more quasi-finite simple modules.

In this paper, we attempt to prove that under a certain condition which is natural from the geometric viewpoint, all quasi-finite simple modules over the Virasoro-like algebra are known. This will be carried out by following a strategy which was invented by E.S. Rao in his partial classification of simple modules over the Lie algebra of vector fields on $d$-dimensional torus (see \cite{rao}). To make Rao's strategy work for the Virasoro-like algebra, we use rather different tricks to fulfill each step. All of the results will be stated and proven in an algebraic way.

The paper is organized as follows. In Section \ref{about-virasoro-like-algebra}, we introduce the Virasoro-like algebra and talk about relevant background. The purpose of Section \ref{general-results} is to gather some general definitions and basic results which are repeatedly used in other sections. Section \ref{backward-induction} is devoted to carry out the inductive strategy, which is the most technical part of this paper. In Section \ref{proof-of-main-result}, after a brief analysis and easy computation, we state the classification theorem and prove it by using the results obtained in previous sections.

Throughout the paper, all vector spaces are over the complex number field $\mathbb C$. Let $\mathbf e_1=(1,0), \mathbf e_2=(0,1)$ be the standard basis of the 2-dimensional vector space $\mathbb C^2$, and $\Gamma=\mathbb Z\mathbf e_1\oplus\mathbb Z\mathbf e_2\subseteq\mathbb C^2$ denote an abelian group isomorphic to $\mathbb Z^2$, where $\mathbb Z$ is the group of integers. The determinant of two vectors $\alpha,\beta\in\mathbb C^2$ will be denoted by $\det(\alpha,\beta)$.

Let $\mathfrak A=\mathbb C[z_1^{\pm1}, z_2^{\pm1}]$ denote the Laurent polynomial algebra in two variables. Given $\mathbf r=r_1\mathbf e_1+r_2\mathbf e_2\in\Gamma$, let $\mathbf z^\mathbf r=z_1^{r_1}z_2^{r_2}\in\mathfrak A$, and $\partial_k=z_k\frac{\partial}{\partial z_k}$ be a derivation on $\mathfrak A$. It's not hard to verify that $\{\mathbf z^\mathbf r \partial_k~|~\mathbf r\in\Gamma, k=1,2\}$ is a basis of the derivation Lie algebra $\operatorname{Der}\mathfrak A$. For any $\mathbf r,\mathbf s\in\Gamma$ and $k,\ell=1,2$, we have
\begin{equation*}
[\mathbf z^\mathbf r\partial_k,\,\mathbf z^\mathbf s\partial_\ell]=\mathbf z^{\mathbf r+\mathbf s}(s_k\partial_\ell-r_\ell\partial_k).
\end{equation*}
Note that $\mathbf z^\mathbf r\partial_k(\mathbf z^\mathbf s)=s_k\mathbf z^{\mathbf r+\mathbf s}$ holds, we may now consider the natural abelian extension $\mathfrak A\rtimes\operatorname{Der}\mathfrak A$ satisfying
\begin{equation*}
[\mathbf z^\mathbf r, \mathbf z^\mathbf s]=0, \quad [\mathbf z^\mathbf r\partial_k, \mathbf z^\mathbf s]=s_k\mathbf z^{\mathbf r+\mathbf s}, \quad \mathbf r,\mathbf s\in\Gamma,~~k=1,2.
\end{equation*}
Then $\mathfrak H=\operatorname{Span}\{1,\partial_1,\partial_2\}$ is a maximal abelian subalgebra of the Lie algebra $\mathfrak A\rtimes\operatorname{Der}\mathfrak A$.

Now, assume that $\{h, e, f\}$ is the Chevalley basis of the 3-dimensional simple Lie algebra $\mathfrak{sl}_2$ satisfying $[h,e]=2e,\, [h,f]=-2f,\, [e,f]=h$. And let $V(\lambda)$ be a finite-dimensional simple module over $\mathfrak{sl}_2$ with highest weight $\lambda\in\mathbb N$, where $\mathbb N$ denotes the set of nonnegative integers. For any $\mu\in\mathbb C$ and $\alpha=(\alpha_1,\alpha_2)\in\mathbb C^2$, we can construct an $\mathfrak A\rtimes\operatorname{Der}\mathfrak A$-module $M_\mu^{\alpha}(\lambda)=\mathfrak A\otimes_\mathbb C V(\lambda)$ as follows. First denote $\mathbf z^\mathbf s\otimes v$ by $|v,\mathbf s\rangle$ for convenience, where $v\in V(\lambda)$ and $\mathbf s\in\Gamma$. Then let $\mathbf z^\mathbf r |v,\mathbf s\rangle = |v,\,\mathbf r+\mathbf s\rangle$ and
\begin{gather*}
\mathbf z^\mathbf r\partial_1 |v,\mathbf s\rangle = \left|\left(s_1+\alpha_1+r_1\mu+r_2f+\frac{r_1}{2}h\right) v,\, \mathbf r+\mathbf s\right\rangle,\\
\mathbf z^\mathbf r\partial_2 |v,\mathbf s\rangle = \left|\left(s_2+\alpha_2+r_2\mu+r_1e-\frac{r_2}{2}h\right) v,\, \mathbf r+\mathbf s\right\rangle.
\end{gather*}
It can be verified by direct computation that $M_\mu^\alpha(\lambda)$ is indeed a module over $\mathfrak A\rtimes\operatorname{Der}\mathfrak A$ and also a weight module over $\mathfrak H$ (cf. \cite{larsson,shen}).

If an $\mathfrak A\rtimes\operatorname{Der}\mathfrak A$-module is not only a weight $\mathfrak H$-module, but also a module over the unital associative algebra $\mathfrak A$, then it's called a $(\mathfrak A,\operatorname{Der}\mathfrak A)$-module. S.E. Rao \cite{rao} proved that $M_\mu^\alpha(\lambda)$ is a simple $(\mathfrak A,\operatorname{Der}\mathfrak A)$-module, and any quasi-finite (the dimension of every weight space is finite) simple $(\mathfrak A,\operatorname{Der}\mathfrak A)$-module coincides with $M_\mu^\alpha(\lambda)$ for some $\alpha,\mu$ and $\lambda$. Recently, Y. Billig and V. Futorny \cite{billig} proved that any cuspidal (the dimensions of weight spaces are uniformly bounded by some constant) simple $\operatorname{Der}\mathfrak A$-module induces a cuspidal $(\mathfrak A,\operatorname{Der}\mathfrak A)$-module. In other words, they reduced the classification of cuspidal $\operatorname{Der}\mathfrak A$-modules to the classification of cuspidal $(\mathfrak A,\operatorname{Der}\mathfrak A)$-modules.

On the other hand, V. Mazorchuk and K. Zhao \cite{mazorchuk} proved that every quasi-finite simple $\operatorname{Der}\mathfrak A$-module is either cuspidal or of highest weight type. They also described supports (sets of weights) of each generalized highest weight simple modules. Thus we obtain a complete list of quasi-finite simple modules over the Lie algebra $\operatorname{Der}\mathfrak A$. In this process, we thus understand the crucial intermediate role played by the classification of cuspidal simple $(\mathfrak A,\operatorname{Der}\mathfrak A)$-modules.

\section{About Virasoro-like algebra}
\label{about-virasoro-like-algebra}

To put things on a geometric setting, we may think of $\mathfrak A$ (resp. $\operatorname{Der}\mathfrak A$) as the pure-algebraic version of the commutative algebra $C^\infty(\mathbb T^2)$ of smooth functions (resp. the Lie algebra $\operatorname{Vect}(\mathbb T^2)$ of smooth vector fields) on the 2-dimensional torus $\mathbb T^2$ under a correspondence
\begin{equation*}
\mathbf z^\mathbf r\longleftrightarrow e^{\imath(r_1\theta_1+r_2\theta_2)}, \quad \partial_k\longleftrightarrow-\imath\frac{\partial}{\partial\theta_k}, \quad \imath=\sqrt{-1},~~ k=1,2.
\end{equation*}
Note that $\mathbb T^2$ is a symplectic manifold hence admits a Poisson structure. The Poisson bracket of two functions $\phi$ and $\psi$ is given by
\begin{equation*}
\{\phi,\,\psi\}:=\frac{\partial\phi}{\partial\theta_2}\frac{\partial\psi}{\partial\theta_1}-
\frac{\partial\phi}{\partial\theta_1}\frac{\partial\psi}{\partial\theta_2}.
\end{equation*}
Accordingly, the Hamiltonian vector field generated by the function $\mathbf z^\mathbf r$ is clearly $H(\mathbf r):=\mathbf z^\mathbf r(r_1\partial_2-r_2\partial_1)$. For any $\mathbf r,\mathbf s\in\Gamma$, we have
\begin{equation*}
[H(\mathbf r), H(\mathbf s)]=\det(\mathbf r, \mathbf s) H(\mathbf r+\mathbf s).
\end{equation*}
The Lie subalgebra $$\mathfrak P:=\operatorname{Span}\{H(\mathbf r)~|~\mathbf r\in\Gamma\}$$ of $\operatorname{Der}\mathfrak A$ is usually referred to as the \emph{Virasoro-like algebra}. It's a special kind of the so-called Block type Lie algebra which was studied by many authors in the past.

Since $H(\mathbf r)\mathbf z^\mathbf s=\det(\mathbf r,\mathbf s)\mathbf z^{\mathbf r+\mathbf s}$, the subalgebra $\mathfrak A\rtimes\mathfrak P$ of $\mathfrak A\rtimes\operatorname{Der}\mathfrak A$ satisfies
\begin{equation*}
[\mathbf z^\mathbf r,\mathbf z^\mathbf s]=0,\quad [H(\mathbf r),\mathbf z^\mathbf s]=\det(\mathbf r,\mathbf s)\mathbf z^{\mathbf r+\mathbf s},\quad \mathbf r,\mathbf s\in\Gamma.
\end{equation*}
And the $(\mathfrak A,\operatorname{Der}\mathfrak A)$-module $M_\mu^\alpha(\lambda)$ is also a $(\mathfrak A,\mathfrak P)$-module by restriciton. An easy computation shows that
\begin{equation*}
H(\mathbf r)|v,\mathbf s\rangle=\left|\left\{\det(\mathbf r,\mathbf s+\alpha)+r_1^2e-r_2^2f-r_1r_2h\right\}v,\,\mathbf r+\mathbf s\right\rangle.
\end{equation*}
Note that as a $(\mathfrak A,\mathfrak P)$-module, $M_\mu^\alpha(\lambda)$ is independent of $\mu$, henceforth it will be denoted by $M^\alpha(\lambda)$.

W. Lin and Y. Su \cite{lin} proved that every quasi-finite $\Gamma$-graded simple $\mathfrak P$-module is either cuspidal or of highest weight type. Apart from this, they also determined all generalized highest weight simple modules. Nonetheless, to the best of our knowledge, we are not aware of any more cuspidal simple $\mathfrak P$-module. Inspired by the work of S.E. Rao, Y. Billig and V. Futorny on $\operatorname{Der}\mathfrak A$-modules, we attempt to prove that all quasi-finite simple $(\mathfrak A,\mathfrak P)$-modules come from $M^\alpha(\lambda)$ for some $\alpha$ and $\lambda$. The additional $\mathfrak A$-module structure on a $\mathfrak P$-module is not merely a technical requirement but a natural geometric condition, since the action of vector fields on the sections of vector bundles forms a $(\mathfrak A,\mathfrak P)$-module structure.

In fact, we have not been able to classify all cuspidal simple $\mathfrak P$-modules by extending Billig-Futorny's reduction method. This makes us conjecture that the reduction method for the $\operatorname{Der}\mathfrak A$-modules does not work for the $\mathfrak P$-modules. A bit of explanation is in order, but let us first recall their main ideas briefly.

Every vector field on a $d$-dimensional torus can be approximated by linear combinations of vector fields on $d$ independent solenoids. Here, a solenoid is by definition, a dense 1-dimensional immersed submanifold of a torus. Thus the Lie algebra $\operatorname{Der}\mathfrak A$ can be decomposed (as a vector space) into a direct sum of $d$ solenoidal subalgebras (which are approximate to Lie algebras of vector fields on solenoids). Since every solenoid is diffeomorphic to the real line $\mathbb R$, the properties of the solenoidal subalgebras are very similar to those of the Virasoro algebra. It is the well-developed representation theory of the Virasoro algebra that makes the reduction feasible.

Now compare them with $\mathfrak P$. Since the Poisson bracket always vanishes on any 1-dimensional manifold, it's kind of hopeless to find any nonabelian ``rank one'' subalgebra of $\mathfrak P$, and since abelian subalgebras impose extremely weak restriction on the structure of $\mathfrak P$-module, the reduction fails. From this viewpoint, if we consider the Virasoro algebra as the first nontrivial example of infinite-dimensional graded Lie algebra of rank one, then we may think of the Virasoro-like algebra $\mathfrak P$ as the first nontrivial example of graded Lie algebra of rank two. This algebra deserves more attention than have been paid in the past. We are particularly eager to see if there exists a complete classification of cuspidal simple $\mathfrak P$-modules.

As to the classification of cuspidal $(\mathfrak A,\mathfrak P)$-modules, we follow S.E. Rao's strategy of ``backward induction'' (cf. \cite{rao}). We first confirm the inductive part: If all of the $k^\text{th}$-order differences vanish, then so do the $(k-1)^\text{th}$-order differences. Then we verify the initial conditions: All $k^\text{th}$-order differences vanish for sufficiently large integer $k$. Combining them we assert that all $3^\text{rd}$-order differences vanish, hence each $H(\mathbf r)$ is represented by a family of quadratic polynomial matrices in two variables $r_1$ and $r_2$. Following this, we will be in a position to deduce that any quasi-finite simple $(\mathfrak A,\mathfrak P)$-module is of the form $M^\alpha(\lambda)$.

\section{Definitions and basic results}
\label{general-results}

The purpose of this section is to introduce some important notations, useful formulae, basic results, and fundamental lemmas. They will be repeatedly used in the subsequent sections.

Let $M=\oplus_{\mathbf s\in\Gamma}M_{\mathbf s}$ be a \emph{quasi-finite $\Gamma$-graded $(\mathfrak A,\mathfrak P)$-module}, namely, for any $\mathbf r,\mathbf s\in\Gamma$, $v\in M$, we have
\begin{gather*}
\dim M_{\mathbf s}<+\infty, ~~ \mathbf z^{\mathbf r}M_{\mathbf s}\subseteq M_{\mathbf{r+s}}, ~~ H(\mathbf r)M_{\mathbf s}\subseteq M_{\mathbf{r+s}}, \\
\mathbf z^\mathbf r(\mathbf z^\mathbf s v)=\mathbf z^\mathbf s(\mathbf z^\mathbf r v)=\mathbf z^{\mathbf{r+s}}v, \quad\quad \mathbf z^\mathbf 0 v=1v=v, \\
[H(\mathbf r),\,\mathbf z^\mathbf s]v=H(\mathbf r)(\mathbf z^\mathbf s v)-\mathbf z^\mathbf s(H(\mathbf r)v)=\det(\mathbf r,\mathbf s)\mathbf z^{\mathbf{r+s}}v, \\
[H(\mathbf r),\,H(\mathbf s)]v=H(\mathbf r)(H(\mathbf s)v)-H(\mathbf s)(H(\mathbf r)v)=\det(\mathbf r,\mathbf s)H(\mathbf{r+s})v.
\end{gather*}
Clearly every $\mathbf z^\mathbf r$ is an isomorphism of $M_\mathbf 0$ into $M_\mathbf r$, thus we need to know how $T(\mathbf r):=\mathbf z^{-\mathbf r}H(\mathbf r)$ acts on $M$. The following commutation relations are in order.
\begin{align*}
[T(\mathbf r),\,\mathbf z^\mathbf s]&=[\mathbf z^{-\mathbf r}H(\mathbf r), \mathbf z^\mathbf s]=\mathbf z^{-\mathbf r}[H(\mathbf r),\mathbf z^\mathbf s]=\det(\mathbf r,\mathbf s)\mathbf z^\mathbf s,\\
[T(\mathbf r),\,H(\mathbf s)]&=[\mathbf z^{-\mathbf r}H(\mathbf r),H(\mathbf s)]=[\mathbf z^{-\mathbf r},H(\mathbf s)]H(\mathbf r)+\mathbf z^{-\mathbf r}[H(\mathbf r),H(\mathbf s)]\\
&=\det(\mathbf s,\mathbf r)\mathbf z^{\mathbf s-\mathbf r}H(\mathbf r)+\det(\mathbf r,\mathbf s)\mathbf z^{-\mathbf r}H(\mathbf{r+s})\\
&=\det(\mathbf r,\mathbf s)\mathbf z^\mathbf s\left\{T(\mathbf{r+s})-T(\mathbf r)\right\},\\
[T(\mathbf r),\,T(\mathbf s)]&=[T(\mathbf r),\mathbf z^{-\mathbf s}H(\mathbf s)]=[T(\mathbf r),\mathbf z^{-\mathbf s}]H(\mathbf s)+\mathbf z^{-\mathbf s}[T(\mathbf r),H(\mathbf s)]\\
&=\det(\mathbf r,\mathbf s)\left\{T(\mathbf{r+s})-T(\mathbf r)-T(\mathbf s)\right\}.
\end{align*}
It follows that
$$\mathfrak T:=\operatorname{Span}\{T(\mathbf r)~|~\mathbf r\in\Gamma\}$$
forms a Lie algebra. We will study its finite-dimensional simple modules. To this end, for any $\mathbf u\in\Gamma$, consider the \emph{$\mathbf u$-difference of $T(\mathbf r)$}, viz.,
$$\Delta_\mathbf u T(\mathbf r):=T(\mathbf r)-T(\mathbf{r+u}).$$
It is easy to see that the following holds.
\begin{align*}
\Delta_k(\mathbf r;\mathbf u_1,\cdots,\mathbf u_k)&:=\Delta_{\mathbf u_1}\cdots\Delta_{\mathbf u_k}T(\mathbf r)=T(\mathbf r)-\sum_{1\le i\le k}T(\mathbf r+\mathbf u_i)+\cdots\\
& \quad\quad +(-1)^\ell\sum_{1\le i_1<\cdots<i_\ell\le k}T(\mathbf r+\mathbf u_{i_1}+\cdots+\mathbf u_{i_\ell})\\
& \quad\quad +\cdots+(-1)^k\,T(\mathbf r+\mathbf u_1+\cdots+\mathbf u_k).
\end{align*}
Hence $\Delta_k(\mathbf r;\mathbf u_1,\cdots,\mathbf u_k)=\Delta_k(\mathbf r;\mathbf u_{\sigma(1)},\cdots,\mathbf u_{\sigma(k)})$ for all $k\ge 1$, where $\sigma$ is any permutation of $\{1,2,\cdots,k\}$. In addition, by the definition of difference we see that
\begin{equation*}
\begin{gathered}
\Delta_k(\mathbf r;-\mathbf u_1,\mathbf u_2,\cdots,\mathbf u_k)=-\Delta_k(\mathbf r-\mathbf u_1;\mathbf u_1,\cdots,\mathbf u_k),\\
\Delta_{k+1}(\mathbf r;\mathbf u_1,\cdots,\mathbf u_{k+1})=\Delta_k(\mathbf r;\mathbf u_1,\cdots,\mathbf u_k)-\Delta_k(\mathbf r+\mathbf u_{k+1};\mathbf u_1,\cdots,\mathbf u_k).
\end{gathered}
\end{equation*}
Besides, an easy induction will give us
\begin{gather*}
\Delta_k(\mathbf r;\mathbf u_1+\mathbf u'_1,\mathbf u_2,\cdots,\mathbf u_k)=\Delta_k(\mathbf r;\mathbf u_1,\mathbf u_2,\cdots,\mathbf u_k)\\
\quad\quad +\Delta_k(\mathbf r+\mathbf u_1;\mathbf u'_1,\mathbf u_2,\cdots,\mathbf u_k).
\end{gather*}
Now we are in a position to claim that, for all $k\ge 2$ and $\mathbf r,\mathbf s,\mathbf u_1,\cdots,\mathbf u_k\in\Gamma$, the following identities hold.
\begin{align*}
&[T(\mathbf r),\,\Delta_k(\mathbf s;\mathbf u_1,\cdots,\mathbf u_k)]=\det(\mathbf s,\mathbf r)\Delta_{k+1}(\mathbf s;\mathbf u_1,\cdots,\mathbf u_k,\mathbf r)\\
&\quad\quad\quad\quad\quad\quad\quad +\sum_{i=1}^k\det(\mathbf r,\mathbf u_i)\Delta_k(\mathbf s+\mathbf u_i;\mathbf u_1,\cdots,\widehat{\mathbf u_i},\cdots,\mathbf u_k,\mathbf r)\\
&\quad\quad\quad\quad\quad~~ =\det(\mathbf s+\mathbf u_1+\cdots+\mathbf u_k,\mathbf r)\Delta_{k+1}(\mathbf s;\mathbf u_1,\cdots,\mathbf u_k,\mathbf r)\\
&\quad\quad\quad\quad\quad\quad\quad +\sum_{i=1}^k\det(\mathbf r,\mathbf u_i)\Delta_k(\mathbf s;\mathbf u_1,\cdots,\widehat{\mathbf u_i},\cdots,\mathbf u_k,\mathbf r).
\end{align*}

\begin{proof}
The second equality follows from the first one. We prove the first equality by induction. For $k=2$, unwrapping $[T(\mathbf r),\,\Delta_2(\mathbf s;\mathbf u_1,\mathbf u_2)]$, we get
\begin{align*}
\text{LHS}&=[T(\mathbf r),\,T(\mathbf s)-T(\mathbf s+\mathbf u_1)-T(\mathbf s+\mathbf u_2)+T(\mathbf s+\mathbf u_1+\mathbf u_2)] \\
&=\det(\mathbf r,\mathbf s)\{T(\mathbf{r+s})-T(\mathbf r)-T(\mathbf s)\} \\
&\quad -\det(\mathbf r,\mathbf s+\mathbf u_1)\{T(\mathbf r+\mathbf s+\mathbf u_1)-T(\mathbf r)-T(\mathbf s+\mathbf u_1)\} \\
&\quad -\det(\mathbf r,\mathbf s+\mathbf u_2)\{T(\mathbf r+\mathbf s+\mathbf u_2)-T(\mathbf r)-T(\mathbf s+\mathbf u_2)\} \\
&\quad +\det(\mathbf r,\mathbf s+\mathbf u_1+\mathbf u_2)\{T(\mathbf r+\mathbf s+\mathbf u_1+\mathbf u_2) -T(\mathbf r)-T(\mathbf s+\mathbf u_1+\mathbf u_2)\}.
\end{align*}
By rearranging and wrapping up, it's easy to obtain the desired result. From now on, we assume the equality holds for some $k\ge 2$, then following a rather tedious but easy calculation one arrives at
\begin{align*}
&[T(\mathbf r),\,\Delta_{k+1}(\mathbf s;\mathbf u_1,\cdots,\mathbf u_{k+1})]=[T(\mathbf r),\,\Delta_k(\mathbf s;\mathbf u_1,\cdots,\mathbf u_k)]\\
&\quad\quad -[T(\mathbf r),\,\Delta_k(\mathbf s+\mathbf u_{k+1};\mathbf u_1,\cdots,\mathbf u_k)]\\
&~~ =\det(\mathbf s,\mathbf r)\Delta_{k+1}(\mathbf s;\mathbf u_1,\cdots,\mathbf u_k,\mathbf r)\\
&\quad\quad +\sum_{i=1}^k\det(\mathbf r,\mathbf u_i)\Delta_k(\mathbf s+\mathbf u_i;\mathbf u_1,\cdots,\widehat{\mathbf u_i},\cdots,\mathbf u_k,\mathbf r)\\
&\quad\quad -\det(\mathbf s+\mathbf u_{k+1},\mathbf r)\Delta_{k+1}(\mathbf s+\mathbf u_{k+1};\mathbf u_1,\cdots,\mathbf u_k,\mathbf r)\\
&\quad\quad -\sum_{i=1}^k\det(\mathbf r,\mathbf u_i)\Delta_k(\mathbf s+\mathbf u_{k+1}+\mathbf u_i;\mathbf u_1,\cdots,\widehat{\mathbf u_i},\cdots,\mathbf u_k,\mathbf r)\\
&~~ =\det(\mathbf s,\mathbf r)\Delta_{k+2}(\mathbf s;\mathbf u_1,\cdots,\mathbf u_{k+1},\mathbf r)\\
&\quad\quad +\det(\mathbf r,\mathbf u_{k+1})\Delta_{k+1}(\mathbf s+\mathbf u_{k+1};\mathbf u_1,\cdots,\mathbf u_k,\widehat{\mathbf u_{k+1}},\mathbf r)\\
&\quad\quad +\sum_{i=1}^k\det(\mathbf r,\mathbf u_i)\Delta_{k+1}(\mathbf s+\mathbf u_i;\mathbf u_1,\cdots,\widehat{\mathbf u_i},\cdots,\mathbf u_{k+1},\mathbf r)\\
&~~ =\det(\mathbf s,\mathbf r)\Delta_{k+2}(\mathbf s;\mathbf u_1,\cdots,\mathbf u_{k+1},\mathbf r)\\
&\quad\quad +\sum_{i=1}^{k+1}\det(\mathbf r,\mathbf u_i)\Delta_{k+1}(\mathbf s+\mathbf u_i;\mathbf u_1,\cdots,\widehat{\mathbf u_i},\cdots,\mathbf u_{k+1},\mathbf r).
\end{align*}
This proves the identities.
\end{proof}

From these results, it is not hard to see that for $k\ge 2$,
$$\mathfrak I_k:=\operatorname{Span}\{\Delta_k(\mathbf r;\mathbf u_1,\cdots,\mathbf u_k)~|~\mathbf r,\mathbf u_1,\cdots,\mathbf u_k\in\Gamma\}$$
is an ideal of $\mathfrak T$, and $\mathfrak I_{k+1}\subseteq \mathfrak I_k$. This chain of ideals is our main object for consideration. We end this section by providing a few useful Lemmas.

\begin{lemma}
\label{nilpotency}
Let $V$ be a finite-dimensional vector space, $D$ and $E$ be two linear transformations on $V$ such that $[D,E]=\delta E$. If $\delta\ne 0$, then $E$ is nilpotent.
\end{lemma}

\begin{proof}
Consider the generalized weight space decomposition of $V$ by $D$, namely, $V=\oplus_{\lambda\in\mathbb C}V_\lambda$, where $V_\lambda=\{v\in V~|~(D-\lambda)^n v=0,\,\exists n\in\mathbb N\}$. One can easily see that $EV_\lambda\subseteq V_{\lambda+\delta}$. Since $\delta\ne 0$ and $\dim V<+\infty$, $E^n V_\lambda=0$ for sufficiently large $n$. In other words, $E$ is nilpotent.
\end{proof}

\begin{lemma}
\label{ideal-vanish}
Let $V$ be a simple $\mathfrak T$-module, and $k\ge 2$. If there exists some nonzero $v\in V$ such that $\mathfrak I_kv=0$, then $\mathfrak I_kV=0$.
\end{lemma}

\begin{proof}
Let $W=\{w\in V~|~\mathfrak I_kw=0\}\ne 0$. Since $\mathfrak I_k$ is an ideal of $\mathfrak T$, for every $w\in W$, we have $\mathfrak I_k(\mathfrak Tw)\subseteq[\mathfrak I_k,\mathfrak T]w+\mathfrak T\mathfrak I_kw=0$, i.e., $W$ is $\mathfrak T$-invariant. By the simplicity of $V$, we know that $W=V$, hence $\mathfrak I_kV=0$.
\end{proof}

\section{Backward induction}
\label{backward-induction}

In this section, we aim to prove the vanishing of higher-order differences $\Delta_k$ ($k\ge 3$) by using a strategy of ``backward induction'', discovered by S.E. Rao \cite{rao} in his study of Lie algebras of vector fields on the tori.

To begin with, let us investigate the effect of $\mathfrak I_{k+1}=0$ on the behaviour of $\mathfrak I_k$. For every fixed integer $k\ge 2$, suppose $\mathfrak I_{k+1}=0$, then the following neatly written equations follow (consult Section \ref{general-results})
\begin{gather*}
\Delta_k(\mathbf r;\mathbf u_1,\cdots,\mathbf u_k)=\Delta_k(\mathbf 0;\mathbf u_1,\cdots,\mathbf u_k),\\
\Delta_k(\mathbf 0;-\mathbf u_1,\mathbf u_2,\cdots,\mathbf u_k)=-\Delta_k(\mathbf 0;\mathbf u_1,\mathbf u_2\cdots,\mathbf u_k),\\
\Delta_k(\mathbf 0;\mathbf u_1+\mathbf u'_1,\mathbf u_2,\cdots,\mathbf u_k)=\Delta_k(\mathbf 0;\mathbf u_1,\cdots,\mathbf u_k)+\Delta_k(\mathbf 0;\mathbf u'_1,\mathbf u_2,\cdots,\mathbf u_k).
\end{gather*}
Since $\mathbf u_1,\cdots,\mathbf u_k\in\Gamma\subseteq\mathbb C^2$, one may identify $\Delta_k(\mathbf r;\mathbf u_1,\cdots,\mathbf u_k)$ with the symmetric tensor $\mathbf u_1\cdots\mathbf u_k\in\mathcal{S}^k(\mathbb C^2)$. Moreover, for every $\mathbf u\in\Gamma$ and integral matrix
$\begin{pmatrix}
  a & b \\
  c & d \\
\end{pmatrix}$, denote $(a u_1+b u_2)\mathbf e_1+(c u_1+d u_2)\mathbf e_2\in\Gamma$ by
$\begin{pmatrix}
  a & b \\
  c & d \\
\end{pmatrix}\mathbf u$, then a simple computation shows that
\begin{align*}
T(\mathbf r)(\mathbf u_1\cdots\mathbf u_k)&:=[T(\mathbf r),\,\Delta_k(\mathbf 0;\mathbf u_1,\cdots,\mathbf u_k)]\\
&=\sum_{i=1}^k\mathbf u_1\cdots\left\{\begin{pmatrix}
                                       -r_1r_2 & r_1^2 \\
                                       -r_2^2 & r_1r_2 \\
                                     \end{pmatrix}\mathbf u_i\right\}\cdots\mathbf u_k.
\end{align*}
In other words, the adjoint action of $\mathfrak T$ on $\mathfrak I_k$ can be thought of as the $k^\text{th}$ symmetric power of the 2-dimensional simple $\mathfrak{sl}_2$-module under the following correspondences
\begin{equation*}
T(\mathbf r)\longleftrightarrow\begin{pmatrix}
   -r_1r_2 & r_1^2 \\
   -r_2^2 & r_1r_2 \\
\end{pmatrix}, \quad~~ \Delta_k(\mathbf r;\mathbf u_1,\cdots,\mathbf u_k)\longleftrightarrow\mathbf u_1\cdots\mathbf u_k.
\end{equation*}
Such perspective will greatly facilitate our computations and arguments.

\begin{lemma}
\label{inductive-part}
Let $V$ be a finite-dimensional simple $\mathfrak T$-module, and for some $k\ge 3$, $\mathfrak I_{k+1}V=0$, then $\mathfrak I_kV=0$.
\end{lemma}

\begin{proof}
Note that the adjoint action of $T(\mathbf r)$ on $\mathfrak I_k$ can be represented by a quadratic polynomial matrix in two variables, say, $r_1$ and $r_2$. Therefore, all $3^\text{rd}$-order differences of $T(\mathbf r)$ vanish, hence $[\mathfrak I_k,\mathfrak I_k]=0$ for all $k\ge 3$. It follows that $\mathfrak S=\mathfrak I_k+\mathbb CT(\mathbf e_2)$ satisfies $\mathfrak S^{(1)}=[\mathfrak S,\mathfrak S]\subseteq\mathfrak I_k$, $\mathfrak S^{(2)}=[\mathfrak S^{(1)},\mathfrak S^{(1)}]=0$, i.e., $\mathfrak S\subseteq\mathfrak{gl}(V)$ is a solvable Lie algebra.

Denote $\Delta^\ell:=\Delta_k(\mathbf 0;\mathbf e_1,\cdots,\mathbf e_1,\mathbf e_2,\cdots,\mathbf e_2)$, where $\mathbf e_1$ appears $\ell$ times. Thanks to Lie's Theorem (see, e.g., \cite{humphreys}), there exists a nonzero $v\in V$ and $\mu_\ell,\mu\in\mathbb C$ such that $\Delta^\ell v=\mu_\ell v$, $T(\mathbf e_2)v=\mu v$. On the other hand, note that the following holds.
$$[T(\mathbf e_1)+T(\mathbf e_2)-T(\mathbf e_1+\mathbf e_2),\,\Delta^\ell]=(2\ell-k)\Delta^\ell.$$
By Lemma \ref{nilpotency} we know that if $2\ell\ne k$, then $\Delta^\ell$ is nilpotent, hence $\mu_\ell=0$. Even if $2\ell=k$, we still have $\Delta^\ell v=0$ by the equation
$$(\ell+1)\Delta^\ell v=[-T(\mathbf e_2),\Delta^{\ell+1}]v=\Delta^{\ell+1}T(\mathbf e_2)v-T(\mathbf e_2)\Delta^{\ell+1}v=0.$$
Since $\{\Delta^\ell~|~0\le\ell\le k\}$ linearly spans $\mathfrak I_k$, we get $\mathfrak I_k v=0$, and Lemma \ref{ideal-vanish} tells us that $\mathfrak I_k V=0$.
\end{proof}

Our next goal is a verification of the ``initial conditions''. To this end, for any $\mathbf r,\mathbf s\in\Gamma$ subject to the relation $\det(\mathbf r,\mathbf s)=\pm1$, let $\Theta$ be the group automorphism of $\Gamma$ which maps $(\mathbf e_1,\mathbf e_2)$ to $(\mathbf r,\mathbf s)$, then the linear mapping
$$\langle\mathbf r,\mathbf s\rangle : T(\mathbf u) \mapsto \det(\mathbf r,\mathbf s)\,T\left(\Theta(\mathbf u)\right)$$
is an automorphism of the Lie algebra $\mathfrak T$. Henceforth we will refer to $\langle\mathbf r,\mathbf s\rangle$ as a \emph{degree-transform of \,$\mathfrak T$}.

\begin{lemma}
\label{initial-condition}
Let $V$ be a finite-dimensional $\mathfrak T$-module. For sufficiently large integer $k$, we have $\mathfrak I_kV=0$.
\end{lemma}

\begin{proof}
Consider the linear transformation
$$X:=-\operatorname{ad}T(-\mathbf e_2)\operatorname{ad}T(-\mathbf e_1)$$
on each ideal $\mathfrak I_k$. Since $\dim\mathfrak I_k\le\dim\mathfrak T\le(\dim V)^2<+\infty$, the eigenvalues of $X$ are bounded. On the other hand, an easy computation indicates that
\begin{equation*}
X\,\Delta_k(\mathbf e_1;\mathbf e_2,\cdots,\mathbf e_2)=k\,\Delta_k(\mathbf e_1;\mathbf e_2,\cdots,\mathbf e_2).
\end{equation*}
So there must exist some $K>3$ such that $\Delta_k(\mathbf e_1;\mathbf e_2,\cdots,\mathbf e_2)=0$ for all integers $k\ge K$. Therefore, for any $r_2\in\mathbb Z$ and $k\ge K$, we have
\begin{align*}
0&=[T(r_2\mathbf e_2),\,\Delta_k(\mathbf e_1;\mathbf e_2,\cdots,\mathbf e_2)]=r_2\,\Delta_{k+1}(\mathbf e_1;r_2\mathbf e_2,\mathbf e_2,\cdots,\mathbf e_2)\\
&=r_2\left\{\Delta_k(\mathbf e_1;\mathbf e_2,\cdots,\mathbf e_2)-\Delta_k(\mathbf e_1+r_2\mathbf e_2;\mathbf e_2,\cdots,\mathbf e_2)\right\}.
\end{align*}
Equivalently, $\Delta_{\mathbf e_2}\cdots\Delta_{\mathbf e_2}T(\mathbf e_1+r_2\mathbf e_2)=0$, with $\Delta_{\mathbf e_2}$ appearing $k$ times. It means that there exists a polynomial matrix $N(r_2)$ in the variable $r_2$ with degree at most $(K-1)$, such that $T(\mathbf e_1+r_2\mathbf e_2)=T(\mathbf e_1)+r_2N(r_2)$. In fact, one can more neatly write the expression by considering
\begin{align*}
r_2[N(r_2),\,T(-\mathbf e_1)]&=[T(\mathbf e_1+r_2\mathbf e_2)-T(\mathbf e_1),\,T(-\mathbf e_1)]\\
&=r_2\left\{T(r_2\mathbf e_2)-T(\mathbf e_1+r_2\mathbf e_2)-T(-\mathbf e_1)\right\}.
\end{align*}
From these relations it is obvious that $T(r_2\mathbf e_2)$ is a polynomial matrix in $r_2$ with degree at most $K$. Since $T(\mathbf 0)=0$, there is a polynomial matrix $P_2(r_2)$ in $r_2$ with degree at most $(K-1)$ such that $T(r_2\mathbf e_2)=r_2P_2(r_2)$.

Now using the degree-transform $\langle\mathbf e_2,\mathbf e_1\rangle$, one may convince oneself that $T(r_1\mathbf e_1)=r_1P_1(r_1)$ holds for a polynomial matrix $P_1(r_1)$ in $r_1$ with degree at most $(K-1)$. Then we have
\begin{align*}
r_1r_2[P_1(r_1),\,P_2(r_2)]&=[T(r_1\mathbf e_1),\,T(r_2\mathbf e_2)]\\
&=r_1r_2\left\{T(\mathbf r)-T(r_1\mathbf e_1)-T(r_2\mathbf e_2)\right\}.
\end{align*}
Hence $T(\mathbf r)=r_1P_1(r_1)+r_2P_2(r_2)+[P_1(r_1),P_2(r_2)]$ is a polynomial matrix in $r_1,r_2$ with degree at most $2K$. Thus for $k>2K$, we always get
$$\Delta_k(\mathbf r;\mathbf u_1,\cdots,\mathbf u_k)=\Delta_{\mathbf u_1}\cdots\Delta_{\mathbf u_k}T(\mathbf r)=0.$$
In other words, for all sufficiently large $k\in\mathbb N$ we have $\mathfrak I_k V=0$.
\end{proof}

\begin{corollary}
\label{vanishing}
If $V$ is a finite-dimensional simple $\mathfrak T$-module, then $\mathfrak I_3 V=0$.
\end{corollary}

\begin{proof}
It follows from Lemma \ref{inductive-part} and \ref{initial-condition} directly.
\end{proof}

\section{Proof of the main result}
\label{proof-of-main-result}

Now we are ready to prove the main theorem. First, let $V$ be a finite-dimensional simple $\mathfrak T$-module, and treat $\mathfrak T$ as a subalgebra of $\mathfrak{gl}(V)$. By Corollary \ref{vanishing}, we know that $\mathfrak I_3=0$. In particular, one may deduce from
$$\Delta_{\mathbf e_1}\Delta_{\mathbf e_1}\Delta_{\mathbf e_1}T(t\,\mathbf e_1)=\Delta_3(t\,\mathbf e_1;\mathbf e_1,\mathbf e_1,\mathbf e_1)=0$$
that there exist linear transformations $\tau_0$, $\tau_1$ and $\tau_2$ on $V$ such that
\begin{equation}
\label{quad-polyn}
T(t\,\mathbf e_1)=\tau_0 + t \tau_1 + t^2 \tau_2.
\end{equation}
Setting $t=0$, we get $\tau_0=0$, and setting $t=1$, we get $\tau_1+\tau_2=T(\mathbf e_1)$, finally setting $t=-1$, we get $\tau_2-\tau_1=T(-\mathbf e_1)$. This gives us
\begin{equation}
\label{tau1-tau2}
\tau_1=\frac{1}{2}\left\{T(\mathbf e_1)-T(-\mathbf e_1)\right\}, \quad~~ \tau_2=\frac{1}{2}\left\{T(\mathbf e_1)+T(-\mathbf e_1)\right\}.
\end{equation}
On the other hand, we have $[\tau_1,\,T(\pm\mathbf e_1)]=0$ and
\begin{gather*}
[T(\mathbf e_1)-T(-\mathbf e_1),\,T(\mathbf e_2)]=\Delta_3(\mathbf 0;\mathbf e_1,-\mathbf e_1,\mathbf e_2)=0.
\end{gather*}
Using the degree-transform $\langle\mathbf e_1,-\mathbf e_2\rangle$, we also have
$$[T(\mathbf e_1)-T(-\mathbf e_1),\,T(-\mathbf e_2)]=0.$$
And it is not hard to see that $\{T(\mathbf e_1),T(-\mathbf e_1),T(\mathbf e_2),T(-\mathbf e_2)\}$ generates $\mathfrak T$, so $[\tau_1,\mathfrak T]=0$. In other words, $\tau_1$ is an endomorphism of the $\mathfrak T$-module $V$. Therefore, Schur's Lemma (see, e.g., \cite{humphreys}) confirms us that $\tau_1$ is a scalar.

It follows from Equation (\ref{quad-polyn}) and (\ref{tau1-tau2}) that
$$T'(r_1\mathbf e_1):=T(r_1\mathbf e_1)-r_1\tau_1=r_1^2 \{T(\mathbf e_1)-\tau_1\}=r_1^2\,T'(\mathbf e_1).$$
Similarly, by using the degree-transform $\langle\mathbf e_2,\mathbf e_1\rangle$, we know that
$$T'(r_2\mathbf e_2):=T(r_2\mathbf e_2)-r_2\tau_2=r_2^2 \{T(\mathbf e_2)-\tau_2\}=r_2^2\,T'(\mathbf e_2),$$
where $\tau_2$ may also be treated as a scalar.
If we denote
$$e:=T'(\mathbf e_1), \quad f:=-T'(\mathbf e_2), \quad h:=T(\mathbf e_1)+T(\mathbf e_2)-T(\mathbf e_1+\mathbf e_2),$$ then a simple calculation shows that
$$T'(\mathbf r) := T(\mathbf r) - r_1\tau_1 - r_2\tau_2 = r_1^2 e - r_2^2 f - r_1 r_2 h.$$
Note $[e,f]=[T(\mathbf e_1),-T(\mathbf e_2)]=T(\mathbf e_1)+T(\mathbf e_2)-T(\mathbf e_1+\mathbf e_2)=h$, and it's easy to obtain $[h,e]=2e$, $[h,f]=-2f$, hence $\{h, e, f\}$ forms the Chevalley basis of $\mathfrak{sl}_2$. Also, note that the Lie algebra $\operatorname{Span}\{T'(\mathbf r)~|~\mathbf r\in\Gamma\}$ is isomorphic to $\mathfrak{sl}_2$, thus $V$ is at the same time, a simple $\mathfrak{sl}_2$-module.

\begin{lemma}
\label{t-module}
If $M=\oplus_{\mathbf s\in\Gamma} M_\mathbf s$ is a $\Gamma$-graded simple $(\mathfrak A,\mathfrak P)$-module, then $M_\mathbf 0$ is a simple $\mathfrak T$-module.
\end{lemma}

\begin{proof}
Suppose $V\subsetneq M_\mathbf 0$ be a nontrivial $\mathfrak T$-submodule, then for $\mathbf r\in\Gamma$, we get $\mathbf z^{-\mathbf r}H(\mathbf r)V\subseteq V$, or equivalently, $H(\mathbf r)V\subseteq\mathbf z^\mathbf rV$. Thus for all $\mathbf s\in\Gamma$, we have $H(\mathbf r)(\mathbf z^\mathbf sV)\subseteq [H(\mathbf r),\mathbf z^\mathbf s]V+\mathbf z^\mathbf sH(\mathbf r)V\subseteq\mathbf z^{\mathbf r+\mathbf s}V$. This means that $M'=\oplus_{\mathbf s\in\Gamma}\mathbf z^{\mathbf s}V$ is a nontrivial submodule of $M$, a contradiction.
\end{proof}

\begin{theorem}
\label{classification-theorem}
If $M=\oplus_{\mathbf s\in\Gamma}M_\mathbf s$ is a quasi-finite simple $(\mathfrak A,\mathfrak P)$-module, then $M\cong M^\alpha(\lambda)$ for some $\alpha\in\mathbb C^2$ and $\lambda\in\mathbb N$.
\end{theorem}

\begin{proof}
By Lemma \ref{t-module}, we know that $M_\mathbf 0$ is a simple $\mathfrak T$-module, hence also a simple $\mathfrak{sl}_2$-module (follows from the previous analysis). Let $\lambda$ be the highest weight of $M_\mathbf 0$ with respect to the Chevalley basis $\{h, e, f\}$ of $\mathfrak{sl}_2$. The previous analysis also shows that for any $v\in M_\mathbf 0$,
\begin{equation*}
H(\mathbf r) v=(r_1\tau_1 + r_2\tau_2)\mathbf z^\mathbf r v + \mathbf z^\mathbf r (r_1^2 e - r_2^2 f - r_1 r_2 h)v,
\end{equation*}
where $\tau_1, \tau_2\in\mathbb C$. Denoting $\mathbf z^\mathbf sv$ by $|v,\mathbf s\rangle$ and noting that $\mathbf z^\mathbf s$ is invertible, one may deduce that, all elements of $M_\mathbf s$ can be represented by $|w,\mathbf s\rangle$ for some $w\in M_\mathbf 0$. Now
\begin{align*}
H(\mathbf r)|v,\mathbf s\rangle&=H(\mathbf r)\mathbf z^\mathbf sv=[H(\mathbf r),\mathbf z^\mathbf s]v+\mathbf z^\mathbf s H(\mathbf r)v\\
&=\left|\left\{\det(\mathbf r,\,\mathbf s+\alpha)+r_1^2e-r_2^2f-r_1r_2h\right\}v,\,\mathbf r+\mathbf s\right\rangle,
\end{align*}
where $\alpha=-\tau_2\mathbf e_1+\tau_1\mathbf e_2\in\mathbb C^2$. This proves that $M\cong M^\alpha(\lambda)$.
\end{proof}

\section*{Acknowledgments}

The first author would like to thank Professor Chengming Bai for great encouragement and Professor Syed Twareque Ali for beneficial discussions. A special thank is due to Doctor Syed Chowdhury who made helpful comments about the draft of this paper.






\end{document}